\documentclass[conference]{IEEEtran}
\IEEEoverridecommandlockouts

\usepackage{cite}
\usepackage{amsmath,amssymb,amsfonts}
\usepackage{algorithm}
\usepackage{algorithmic}
\usepackage{graphicx}
\usepackage{textcomp}
\usepackage{xcolor}
\usepackage{booktabs}
\usepackage{multirow}
\usepackage{tikz}
\usetikzlibrary{arrows.meta,positioning,fit,backgrounds}
\usepackage{array}
\usepackage[caption=false,font=footnotesize]{subfig}

\usepackage{amsthm}
\def\BibTeX{{\rm B\kern-.05em{\sc i\kern-.025em b}\kern-.08em
    T\kern-.1667em\lower.7ex\hbox{E}\kern-.125emX}}

\newtheoremstyle{ieeeplain}{6pt}{6pt}{\itshape}{}{\bfseries}{.}{ }{}
\newtheoremstyle{ieeedef}{6pt}{6pt}{\upshape}{}{\bfseries}{.}{ }{}
\theoremstyle{ieeeplain}
\newtheorem{theorem}{Theorem}
\newtheorem{lemma}[theorem]{Lemma}
\newtheorem{proposition}[theorem]{Proposition}

\theoremstyle{ieeedef}
\newtheorem{definition}[theorem]{Definition}

\newtheorem*{example}{Example}

\newcommand{\method}{ExactCert}
\newcommand{\Qbox}{\mathcal{Q}}                 
\newcommand{\Gal}{\mathcal{G}}                  
\newcommand{\Win}{\mathcal{W}_{\tau}}            
\newcommand{\Crit}{\mathcal{S}}                 
\newcommand{\Dom}{\mathcal{B}}                  
\newcommand{\wmin}[3]{m_{#1}(#2,#3)}            
\newcommand{\dist}{d_{\tau}}
\newcommand{\gap}{\Gamma}
\newcommand{\Brel}{B^{\mathrm{rel}}}
\newcommand{\Uind}{U^{\mathrm{ind}}}
\newcommand{\xup}{x^{\uparrow}}

\begin{document}

\title{When Imputation Changes the Nearest Neighbor: Exact Certification of Similarity Retrieval for Incomplete Sensor Time Series}

\author{\IEEEauthorblockN{Jianru Shen}
\IEEEauthorblockA{\textit{Columbia University}\\
New York, NY, USA \\
js5929@columbia.edu}
}

\maketitle

\begin{abstract}
Sensor time series routinely reach similarity search with missing values, and the retrieved neighbours then depend on how the gaps were filled. We ask which retrieval conclusions hold for every admissible completion of an incomplete query. For a family of window-minimum distances that contains the L1 distance as its separable end we give a pairwise certificate that one candidate is strictly closer than another, a guaranteed candidate set that contains every possible nearest neighbour, and a certificate that a nearest-neighbour label cannot change. The guarantees hold uniformly over the specified box and apply to the true query whenever it lies inside. At the separable end a closed form is known, but beyond it the worst-case gap need not be attained at a corner of the box. We prove a finite reduction to critical values, give an exact interval dynamic programme for the worst-case ranking gap, and a quadratic-time fast layer of valid bounds and feasible witnesses, supported by an exact criterion for when the feature-box relaxation is tight. On three traffic-sensor collections with real gap patterns, linear interpolation changes the nearest neighbour in 15 to 33 percent of the queries. At the full band the feature-box relaxation misses up to 14 percent of the exactly stable boundary pairs, whereas fixed-completion and Monte Carlo checks issue false certificates. The fast layer resolves about 90 percent of all tested boundary pairs without dynamic programming, and on paired network-scale queries the exact method is 82 to 110 times faster than the programme alone. Label certificates transfer to three public benchmarks, and the band width trades forecasting accuracy and certification coverage for lower imputation sensitivity, with dataset-dependent crossovers.
\end{abstract}

\begin{IEEEkeywords}
similarity search, missing values, time series, certified retrieval, dynamic programming, sensor data
\end{IEEEkeywords}

\section{Introduction}\label{sec:intro}

A sensor produces one series per day, and a common way to use such data is to look up similar days. An operator asks which past afternoons looked like today on this loop detector, a forecaster averages the hours that followed those afternoons, and a classifier takes the label of the closest day. When today has gaps, practice fills them by some rule and then ranks the archive as if the data were complete. The ranking then belongs to the rule as much as to the data. On three public traffic-sensor collections with real gap patterns, filling the gaps by linear interpolation returns a nearest neighbour different from the one the true series would return for 15\% to 33\% of the queries, and the historical mean does so for up to 44\%.

This paper asks a question that does not depend on any rule. Given only bounds on the missing values, which retrieval conclusions hold no matter how the gaps are filled. We answer it for a family of window-minimum distances, which compare two series through the minimum of every time span up to a chosen width. One end of the family is the L1 distance, where every coordinate is compared on its own. The other end compares the deepest congestion in every span, a natural quantity for speed data. In between, a band width trades one for the other. We do not claim that this family beats L1 or warping distances. It carries a window semantics that some users want, and it is a simple nonseparable family in which exact certification is needed, because the worst completion need not lie at a corner of the box. The unit of every guarantee is the worst-case ranking gap of a pair of candidates over the box, certified when negative, and the two-slot example in Fig.~\ref{fig:overview} shows a worst-case gap not attained at a corner, which corner-only checks miss. Section~\ref{sec:experiments} shows where the augmented fixed-completion and Monte Carlo baselines fail.

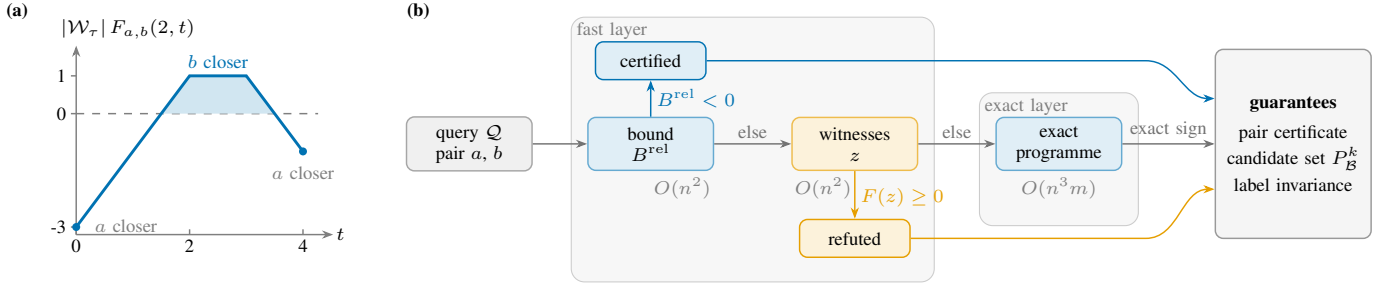
\begin{figure*}[!t]
\centering
\definecolor{cblue}{HTML}{0072B2}
\definecolor{corange}{HTML}{E69F00}
\definecolor{cgrey}{HTML}{7F7F7F}
\begin{tikzpicture}[font=\scriptsize, >={Stealth[length=1.5mm, width=1.1mm]}, line width=0.55pt]
\node[anchor=west] at (-0.25,2.95) {\textbf{(a)}};
\draw[->, cgrey] (0.8,0.0) -- (0.8,2.5);
\draw[->, cgrey] (0.8,0.0) -- (4.2,0.0) node[anchor=west, inner sep=1pt, black] {$t$};
\node[anchor=south west, inner sep=1pt] at (0.55,2.55) {$|\mathcal{W}_{\tau}|\,F_{a,b}(2,t)$};
\foreach \t/\x in {0/0.8, 2/2.3, 4/3.8} { \draw[cgrey] (\x,0.0) -- (\x,0.07); \node[below, inner sep=1.5pt] at (\x,0.0) {\t}; }
\foreach \v/\y in {-3/0.1, 0/1.6, 1/2.1} { \draw[cgrey] (0.73,\y) -- (0.8,\y); \node[left, inner sep=1.5pt] at (0.73,\y) {\v}; }
\draw[dashed, cgrey] (0.8,1.6) -- (4.0,1.6);
\fill[cblue!18] (1.925,1.6) -- (2.3,2.1) -- (3.05,2.1) -- (3.425,1.6) -- cycle;
\draw[cblue, line width=1.1pt] (0.8,0.1) -- (2.3,2.1) -- (3.05,2.1) -- (3.8,1.1);
\fill[cblue] (0.8,0.1) circle (1.5pt);
\fill[cblue] (3.8,1.1) circle (1.5pt);
\node[cblue, anchor=south, inner sep=1pt] at (2.675,2.15) {$b$ closer};
\node[cgrey, anchor=west, inner sep=1pt] at (1.0,0.12) {$a$ closer};
\node[cgrey, anchor=north, inner sep=1pt] at (3.8,0.95) {$a$ closer};
\begin{scope}[xshift=5.2cm]
\node[anchor=west] at (-0.15,2.95) {\textbf{(b)}};
\tikzset{stage/.style={rounded corners=3pt, inner sep=3pt, text width=1.45cm, align=center, minimum height=0.7cm},
         cost/.style={cgrey, inner sep=1pt},
         lab/.style={inner sep=1pt}}
\node[stage, fill=black!6, draw=black!35]     (q)  at (0.8,1.2)  {query $\Qbox$ \\ pair $a$, $b$};
\node[stage, fill=cblue!12, draw=cblue!60]    (bd) at (3.2,1.2)  {bound \\ $\Brel$};
\node[stage, fill=corange!14, draw=corange!70](wt) at (5.9,1.2)  {witnesses \\ $z$};
\node[stage, fill=cblue!12, draw=cblue!60]    (dp) at (8.6,1.2)  {exact \\ programme};
\node[rounded corners=3pt, inner sep=3pt, text width=1.85cm, align=center, minimum height=2.5cm, fill=black!4, draw=black!40] (out) at (11.7,1.2)
      {\textbf{guarantees} \\[3pt] pair certificate \\[2pt] candidate set $P^{k}_{\Dom}$ \\[2pt] label invariance};
\node[stage, fill=cblue!12, draw=cblue, text width=1.25cm, minimum height=0.5cm]   (ce) at (3.2,2.3)  {certified};
\node[stage, fill=corange!14, draw=corange, text width=1.25cm, minimum height=0.5cm](re) at (5.9,-0.05) {refuted};
\node[cost, anchor=north east] at ([yshift=-1pt]bd.south east) (cbd) {$O(n^{2})$};
\node[cost, anchor=north west] at ([yshift=-1pt]wt.south west) (cwt) {$O(n^{2})$};
\node[cost, anchor=north] at ([yshift=-1pt]dp.south) (cdp) {$O(n^{3}m)$};
\begin{scope}[on background layer]
\node[fit=(bd)(wt)(ce)(re)(cbd)(cwt), rounded corners=5pt, inner xsep=6pt, inner ysep=9pt, fill=black!3, draw=black!20] (fast) {};
\node[fit=(dp)(cdp), rounded corners=5pt, inner xsep=6pt, inner ysep=9pt, fill=black!3, draw=black!20] (exact) {};
\end{scope}
\node[cgrey, anchor=north west, inner sep=2pt] at (fast.north west) {fast layer};
\node[cgrey, anchor=north west, inner sep=2pt] at (exact.north west) {exact layer};
\draw[->, black!55] (q) -- (bd);
\draw[->, black!55] (bd) -- node[lab, above=1pt] {else} (wt);
\draw[->, black!55] (wt) -- node[lab, above=1pt] {else} (dp);
\draw[->, cblue]    (bd) -- node[right=2pt, inner sep=0pt, cblue] {$\Brel<0$} (ce);
\draw[->, corange]  (wt) -- node[right=2pt, inner sep=0pt, corange, pos=0.55] {$F(z)\ge 0$} (re);
\draw[->, cblue]    (ce.east) -- (9.75,2.3) to[out=0, in=180] (out.west |- 0,1.8);
\draw[->, corange]  (re.east) -- (9.75,-0.05) to[out=0, in=180] (out.west |- 0,0.6);
\draw[->, black!55] (dp) -- node[lab, above=1pt] {exact sign} (out.west |- dp.east);
\end{scope}
\end{tikzpicture}
\caption{Overview. (a) The two-slot example of Section~\ref{sec:exact}. Filled circles are the two corners of the box, where $a$ is closer, but every completion with $t$ in the shaded range makes $b$ closer, so checking only the corners is insufficient. (b) \method{}. The fast layer certifies a pair from a cheap bound or refutes it from a few witness completions, the exact layer runs the dynamic programme on the pairs that remain, and every outcome is an exact decision of $\gap<0$. The pairwise verdicts assemble into the derived guarantees.}
\label{fig:overview}
\end{figure*}

Our goal is exact certification for a distance family with overlapping-window semantics, and the guarantees hold uniformly over the specified uncertainty box and apply to the true query whenever it lies in that box. In summary, this paper makes three contributions.
\begin{itemize}
\item \textbf{Exact worst-case analysis.} We prove a finite reduction to critical values and derive an exact interval dynamic programme that computes the worst-case gap in $O(n^{3}m)$ time and $O(n^{2})$ space for the value, where $m\le 4n$ is the number of critical values.
\item \textbf{Fast and complete pairwise certification.} We characterise exactly when the feature-box relaxation is tight and construct a witness completion in $O(n^{2})$ time. Combining the bound, witness checks and the programme gives \method{}, which decides whether $\gap(a,b)<0$ for every pair and supports sound candidate sets and sufficient label-invariance certificates.
\item \textbf{Evaluation on six datasets.} We quantify imputation-induced retrieval changes, certification quality, computational cost and band-width trade-offs on three traffic datasets with real gap patterns and three UCR benchmarks.
\end{itemize}

\section{Related Work}\label{sec:related}

\textbf{Uncertain objects and spatial domination.} For objects approximated by rectangles, a necessary and sufficient domination test under $L_{p}$ distances costs linear time in the dimension \cite{emrich2019domination}, and Proposition~\ref{prop:l1} is its instance for two point candidates and a box query under L1. That test relies on separability, and the present paper is about the distances where separability fails. For uncertain time series the usual semantics is probabilistic, with queries answered by probability thresholds or expected distances \cite{assfalg2009probabilistic,dallachiesa2012uncertain,dallachiesa2014topk}, whereas our guarantees hold for every completion of a box.

\textbf{Certain answers and certain predictions.} Certain answers over incomplete databases return what holds in every possible world \cite{imielinski1984incomplete}. Certain predictions carry the idea to nearest neighbour classifiers with uncertain training data \cite{karlas2021certain}, and Fan and Koutris \cite{fan2022certifiable} analyse robustness over subset repairs under functional dependencies, with a dichotomy between polynomial and coNP-hard cases. Both keep the test query fixed. Here the gallery is complete and one uncertain query is shared by all distances, with further coupling from the overlapping window minima.

\textbf{Min-max regret with interval data.} The worst-case gap $\gap(a,b)$ is the maximum regret of choosing $a$ over $b$ when the data lie in a box. Min-max regret optimisation studies such worst cases for combinatorial problems with interval costs \cite{kouvelis1997robust,aissi2009minmax}. For sum objectives the worst scenario sits at interval endpoints and several bottleneck problems are polynomial, while shortest paths and spanning trees become hard \cite{averbakh2004interval}. The example of Section~\ref{sec:exact} shows that endpoints do not suffice for window minima, and Theorem~\ref{thm:dp} shows that the problem is nevertheless polynomial on a path.

\textbf{Threshold decomposition and guaranteed mining under gaps.} Window minima are flat erosions, and the proof of Lemma~\ref{lem:finite} uses the fact that they commute with monotone maps, which is threshold decomposition in the theory of median and stack filters \cite{fitch1984threshold,wendt1986stack}. The mean absolute error of stack filters decomposes over thresholds in the same way \cite{coyle1988stack}. The new element here is the dynamic programme that handles the coupling between thresholds and windows. Guaranteed mining under missing data has been studied for motif discovery, where an admissible bound produces no false negatives \cite{zhu2021admissible}. Our candidate sets are admissible in the same sense for retrieval.

\textbf{Practice with missing sensor data.} Practitioners fill gaps by interpolation, carry forward, historical means, nearest neighbour imputation \cite{troyanskaya2001missing} or learned imputers such as BRITS and SAITS \cite{cao2018brits,du2023saits}, and then retrieve as if the data were complete. Similar-day retrieval itself is a long-standing nonparametric forecaster for traffic \cite{davis1991nonparametric,smith1997traffic}, and tuned dynamic time warping is the usual accuracy reference \cite{sakoe1978dynamic,ratanamahatana2005three}. The paper does not propose an imputer or a forecaster. It measures how much the retrieved answer depends on the imputer and certifies the part that does not.

\section{Problem and Guarantees}\label{sec:problem}

This section fixes the setting, the distance family and the three guarantees that the rest of the paper computes. Table~\ref{tab:notation} collects the main notation.

\begin{table}[t]
\caption{Main notation}
\label{tab:notation}
\centering
\scriptsize
\setlength{\tabcolsep}{3pt}
\begin{tabular}{@{}ll@{}}
\toprule
Symbol & Meaning \\
\midrule
$\Qbox=\prod_{t}[\ell_{t},u_{t}]$ & uncertainty box of the query, completions $x\in\Qbox$ \\
$\Win$, $\wmin{x}{i}{j}$ & windows with $j-i\le\tau$ and the window minimum \\
$\dist(x,a)$, $F_{a,b}(x)$ & window-minimum distance and the ranking gap \\
$\gap(a,b)$ & worst-case gap over $\Qbox$, stable when negative \\
$\alpha_{ij}$, $\beta_{ij}$, $\varphi_{ij}$ & candidate window minima and the per-window term \\
$\Crit$, $s_{1}<\dots<s_{m}$ & critical values, $m\le 4n$ \\
$\underline{z}_{ij}$, $\overline{z}_{ij}$ & feature box of window $(i,j)$ \\
$\Uind$, $\Brel$, $\xup$ & independent bound, relaxation bound, witness \\
$\Dom$, $P^{k}_{\Dom}$ & dominators and the guaranteed candidate set \\
\bottomrule
\end{tabular}
\end{table}

\subsection{Setting}

A query is a sensor series of length $n$ with missing slots, and a gallery $\Gal$ holds $N$ complete series of the same length, for example the earlier days of the same sensor. We ask which conclusions of nearest-neighbour retrieval hold for every admissible way of filling the gaps.

\begin{definition}[Uncertain query and completions]\label{def:box}
An uncertain query is a box $\Qbox=\prod_{t=1}^{n}[\ell_{t},u_{t}]$ with $\ell_{t}\le u_{t}$. An observed slot has $\ell_{t}=u_{t}=x_{t}$ and a missing slot has a nontrivial interval. A completion is any $x\in\Qbox$.
\end{definition}

The box is not a probability model. No completion is preferred, no smoothness across slots is assumed, and every guarantee below holds for all completions at once and is conditional on the true series lying in $\Qbox$.

\subsection{Distance family}

\begin{definition}[Window-minimum distance]\label{def:distance}
For a band width $\tau\in\{0,\dots,n-1\}$ let $\Win$ be the set of windows $(i,j)$ with $1\le i\le j\le n$ and $j-i\le\tau$, and let $\wmin{x}{i}{j}=\min_{i\le t\le j}x_{t}$. The distance between a query $x$ and a candidate $a$ is
\begin{equation}
\dist(x,a)=\frac{1}{|\Win|}\sum_{(i,j)\in\Win}\bigl|\wmin{x}{i}{j}-\wmin{a}{i}{j}\bigr| .
\label{eq:distance}
\end{equation}
\end{definition}

The band width interpolates between two familiar ends. At $\tau=0$ the windows are the single slots and \eqref{eq:distance} is the L1 distance divided by $n$. At $\tau=n-1$ every window counts and the distance compares the full multiscale profile of minima, which for speed data is the deepest congestion inside every time span. The normalisation by $|\Win|$ does not affect any ranking or sign below and is dropped in the implementation, which works on integer codes.

\subsection{Guarantees}

For two candidates $a$ and $b$ write $\alpha_{ij}=\wmin{a}{i}{j}$, $\beta_{ij}=\wmin{b}{i}{j}$ and $\varphi_{ij}(v)=|v-\alpha_{ij}|-|v-\beta_{ij}|$, so that
\begin{equation}
F_{a,b}(x)=\dist(x,a)-\dist(x,b)=\frac{1}{|\Win|}\sum_{(i,j)\in\Win}\varphi_{ij}\bigl(\wmin{x}{i}{j}\bigr).
\label{eq:gap}
\end{equation}

\begin{definition}[Certificates]\label{def:certificates}
The worst-case ranking gap of the ordered pair $(a,b)$ over $\Qbox$ is
\begin{equation}
\gap(a,b)=\max_{x\in\Qbox}F_{a,b}(x).
\label{eq:Gamma}
\end{equation}
The pair is stable when $\gap(a,b)<0$, since then $a$ is strictly closer than $b$ for every completion. When $\gap(a,b)=0$ the candidate $a$ is never farther but some completion ties them, and when $\gap(a,b)>0$ some completion makes $b$ closer. Candidate $b$ dominates candidate $a$ when $\gap(b,a)<0$. For a set of dominators $\Dom\subseteq\Gal$ and an integer $k\ge 1$ the guaranteed candidate set is
\begin{equation}
P^{k}_{\Dom}=\bigl\{\,a\in\Gal:\ |\{b\in\Dom\setminus\{a\}:\gap(b,a)<0\}|<k\,\bigr\}.
\label{eq:candidate-set}
\end{equation}
When candidates carry labels, the label $y$ is certified invariant if some candidate $a$ with label $y$ satisfies $\gap(a,b)<0$ for every candidate $b$ whose label differs from $y$.
\end{definition}

The maximum in \eqref{eq:Gamma} exists because $F_{a,b}$ is continuous and $\Qbox$ is compact. Two consequences are used throughout. Any feasible completion $z\in\Qbox$ gives the lower bound $\gap(a,b)\ge F_{a,b}(z)$, so a single completion with $F_{a,b}(z)\ge 0$ refutes stability. Any upper bound on $\gap(a,b)$ that is negative certifies stability. Algorithms that only ever certify from valid upper bounds are sound, and an algorithm that also decides exactly whether $\gap<0$ is complete.

The guarantees rest on four assumptions. A1, every candidate in $\Gal$ is complete. A2, the box is bounded and closed and contains the true query. A3, the band width $\tau$ is fixed and nothing is renormalised after completion. A4, slots are compared at fixed time positions, so no time warping is allowed. A1 and A4 are limitations that Section~\ref{sec:limitations} returns to.

\subsection{The separable end and the two derived guarantees}

\begin{proposition}[Closed form at $\tau=0$]\label{prop:l1}
For $\tau=0$,
\begin{equation}
\gap(a,b)=\frac{1}{n}\sum_{t=1}^{n}\max\{\varphi_{tt}(\ell_{t}),\varphi_{tt}(u_{t})\},
\label{eq:l1-closed}
\end{equation}
which is computable in $O(n)$ time.
\end{proposition}
\begin{proof}
At $\tau=0$ the gap \eqref{eq:gap} is the sum of the terms $\varphi_{tt}(x_{t})$, each depending on one coordinate. The function $\varphi_{tt}(v)=|v-a_{t}|-|v-b_{t}|$ is piecewise linear with breakpoints at $a_{t}$ and $b_{t}$, constant outside the interval between them and monotone inside it, nondecreasing when $a_{t}\le b_{t}$ and nonincreasing otherwise. A monotone function on an interval attains its maximum at an endpoint, and the coordinates can be chosen independently.
\end{proof}

The same endpoint argument covers the squared Euclidean distance, whose gap is linear in each coordinate. Coordinate-wise reasoning of this kind underlies domination tests for uncertain objects under such distances \cite{emrich2019domination}. Nothing of the sort holds for $\tau>0$, where a single coordinate enters up to $(\tau+1)(\tau+2)/2$ overlapping windows and the maximum of \eqref{eq:gap} need not be attained at a corner of the box. Section~\ref{sec:exact} starts from such an example.

\begin{proposition}[Guaranteed candidate sets]\label{prop:candidate}
Let $x\in\Qbox$ and let $T\subseteq\Gal$ be any top-$k$ set of $x$, that is, a set of $k$ candidates such that no candidate outside $T$ is strictly closer to $x$ than a candidate inside $T$. Then $T\subseteq P^{k}_{\Dom}$ for every $\Dom\subseteq\Gal$. The inclusion still holds when $\gap$ in \eqref{eq:candidate-set} is replaced by any upper bound on $\gap$.
\end{proposition}
\begin{proof}
Let $a\in T$. Every candidate strictly closer to $x$ than $a$ belongs to $T$, so at most $k-1$ candidates are strictly closer than $a$. If $\gap(b,a)<0$ then $b$ is strictly closer than $a$ under every completion, in particular under $x$. Hence fewer than $k$ dominators satisfy $\gap(b,a)<0$ and $a\in P^{k}_{\Dom}$. Replacing $\gap$ by an upper bound can only remove certified dominations, which enlarges the set.
\end{proof}

For $k=1$ the set $P^{1}_{\Dom}$ contains the true nearest neighbour of every completion, and $|P^{1}_{\Dom}|=1$ certifies the nearest neighbour outright. A test that declares dominations that do not hold, for instance because it inspects only a sample of completions, can exclude the true nearest neighbour. Section~\ref{sec:experiments} measures how often this happens.

\begin{proposition}[Label invariance]\label{prop:label}
If the label $y$ is certified invariant in the sense of Definition~\ref{def:certificates}, then for every completion $x\in\Qbox$ every nearest neighbour of $x$ in $\Gal$ carries the label $y$.
\end{proposition}
\begin{proof}
Fix $x\in\Qbox$ and let $g$ be a nearest neighbour of $x$. If the label of $g$ differed from $y$, the certifying candidate $a$ would satisfy $\dist(x,a)<\dist(x,g)$, contradicting the choice of $g$.
\end{proof}

The certificate says that the prediction cannot change, not that it is correct, and it is sufficient rather than necessary. It connects the pairwise gap to certain predictions for nearest-neighbour classifiers over incomplete data \cite{karlas2021certain}, where the training tuples are uncertain and the query is fixed. Here the query is uncertain, which couples the distances to all candidates through one completion and couples the coordinates through the window minima.

\section{Exact Certification}\label{sec:exact}

Two facts drive this section. For $\tau>0$ the worst completion need not sit at a corner of the box, so endpoint checks are not enough, and yet only finitely many completions matter, which is what makes an exact algorithm possible.

\begin{example}
Take $n=2$, $\tau=1$, a query $x=(2,t)$ with $t\in[0,4]$, and candidates $a=(0,4)$ and $b=(5,3)$. The three windows give $|\Win|\,F_{a,b}(2,t)=-1+\varphi_{22}(t)+\varphi_{12}(\min\{2,t\})$, which equals $-3$ at $t=0$, $1$ for every $t\in[2,3]$ and $-1$ at $t=4$. Both endpoints of the box make $a$ closer, but the completions with $t\in[2,3]$ make $b$ closer, so $\gap(a,b)=1/3$ and the pair is not stable. The coupling comes from the window $(1,2)$, whose minimum saturates at $2$ once $t\ge 2$ while the window $(2,2)$ keeps moving.
\end{example}

\subsection{Finite reduction}

Let $\Crit=\{\ell_{t},u_{t},a_{t},b_{t}:t\in[n]\}$ be the set of critical values and let $m=|\Crit|\le 4n$, with sorted elements $s_{1}<\dots<s_{m}$.

\begin{lemma}[Finite reduction]\label{lem:finite}
There is a maximiser $x^{\ast}$ of $F_{a,b}$ over $\Qbox$ with $x^{\ast}_{t}\in\Crit$ for every $t\in[n]$.
\end{lemma}
\begin{proof}
Let $x\in\Qbox$ and suppose some coordinates lie strictly between two consecutive critical values, $s_{k}<x_{t}<s_{k+1}$. Since $\ell_{t}$ and $u_{t}$ belong to $\Crit$ and $\ell_{t}\le x_{t}\le u_{t}$, we have $\ell_{t}\le s_{k}$ and $u_{t}\ge s_{k+1}$, so the whole interval $[s_{k},s_{k+1}]$ is feasible for each such coordinate. Draw $\theta$ uniformly from $(s_{k},s_{k+1})$ and let $\psi_{\theta}$ map a value $v$ to $s_{k}$ if $s_{k}<v<\theta$, to $s_{k+1}$ if $\theta\le v<s_{k+1}$, and to $v$ otherwise. The map $\psi_{\theta}$ is nondecreasing on $\mathbb{R}$, so it commutes with minima, $\wmin{\psi_{\theta}(x)}{i}{j}=\psi_{\theta}(\wmin{x}{i}{j})$ for every window, and $\psi_{\theta}(x)\in\Qbox$.

Every breakpoint $\alpha_{ij}$ or $\beta_{ij}$ of a term $\varphi_{ij}$ is a window minimum of $a$ or $b$ and hence a coordinate of $a$ or $b$, so it lies in $\Crit$ and $\varphi_{ij}$ is affine on $[s_{k},s_{k+1}]$. For a value $v$ in the open interval, $\psi_{\theta}(v)=s_{k+1}$ with probability $(v-s_{k})/(s_{k+1}-s_{k})$ and $\psi_{\theta}(v)=s_{k}$ otherwise, so $\mathbb{E}[\psi_{\theta}(v)]=v$ and, by affinity, $\mathbb{E}[\varphi_{ij}(\psi_{\theta}(v))]=\varphi_{ij}(v)$. For $v$ outside the open interval $\psi_{\theta}(v)=v$. Summing over windows gives $\mathbb{E}[F_{a,b}(\psi_{\theta}(x))]=F_{a,b}(x)$, so some $\theta$ satisfies $F_{a,b}(\psi_{\theta}(x))\ge F_{a,b}(x)$. The rounded completion has no coordinate in $(s_{k},s_{k+1})$ and no coordinate has left $\Crit$. Repeating the step for each of the $m-1$ open intervals produces a completion on the grid $\Crit^{n}$ whose value is at least $F_{a,b}(x)$. Applying this to a maximiser, which exists by compactness, proves the claim.
\end{proof}

The argument needs only that window minima commute with nondecreasing maps and that every $\varphi_{ij}$ is piecewise linear with breakpoints in $\Crit$, so it also covers window maxima and other order statistics. The dynamic programme below is the part that needs minima over contiguous windows.

\subsection{Exact dynamic programme}

By Lemma~\ref{lem:finite} the coordinates may be restricted to $\Crit$. The programme scans the critical values from the largest downwards and, inside each interval of positions, decides where the minimum sits. For $1\le l\le r\le n$ and $1\le k\le m$ let $D(l,r,k)$ be the largest value of
\begin{equation}
\sum_{(i,j)\in\Win,\ l\le i\le j\le r}\varphi_{ij}\bigl(\wmin{x}{i}{j}\bigr)
\label{eq:local-objective}
\end{equation}
over assignments of the coordinates $x_{l},\dots,x_{r}$ with $x_{t}\in\Crit\cap[\ell_{t},u_{t}]$ and $x_{t}\ge s_{k}$ for every $t$, with $D(l,r,k)=-\infty$ when no such assignment exists, $D(l,l-1,k)=0$ for the empty interval and $D(l,r,m+1)=-\infty$. Let
\begin{equation}
C(l,r,p,k)=\sum_{(i,j)\in\Win,\ l\le i\le p\le j\le r}\varphi_{ij}(s_{k})
\label{eq:crossing}
\end{equation}
be the contribution of the windows inside $[l,r]$ that contain the position $p$ when all of them have minimum $s_{k}$. Let $A(l,r,k)=\{p:\ l\le p\le r,\ \ell_{p}\le s_{k}\le u_{p}\}$ be the positions where the value $s_{k}$ is admissible.

\begin{theorem}[Exact worst-case gap]\label{thm:dp}
For all $1\le l\le r\le n$ and $1\le k\le m$,
\begin{align}
D(l,r,k)&=\max\{D(l,r,k+1),\,E(l,r,k)\},\label{eq:dp-D}\\
E(l,r,k)&=\max_{p\in A(l,r,k)}\bigl[D(l,p-1,k)+D(p+1,r,k)\nonumber\\
&\qquad\qquad\qquad\ +C(l,r,p,k)\bigr],\label{eq:dp-E}
\end{align}
with $E(l,r,k)=-\infty$ when no position $p$ is admissible, and $|\Win|\,\gap(a,b)=D(1,n,1)$. The values can be computed in $O(n^{3}m)$ time and $O(n^{2})$ space.
\end{theorem}
\begin{proof}
Fix $(l,r,k)$ and split the admissible assignments by their minimum over $[l,r]$. If the minimum exceeds $s_{k}$, every coordinate is at least $s_{k+1}$ and the assignment is admissible for level $k+1$, so the best such value is $D(l,r,k+1)$. Otherwise the minimum equals $s_{k}$ and is attained at some position $p$, which requires $\ell_{p}\le s_{k}\le u_{p}$. Every window inside $[l,r]$ that contains $p$ then has minimum exactly $s_{k}$, because all coordinates are at least $s_{k}$ and $x_{p}=s_{k}$, and these windows contribute $C(l,r,p,k)$. The remaining windows lie inside $[l,p-1]$ or inside $[p+1,r]$, and the coordinates there form admissible assignments for the same level $k$, since they are at least $s_{k}$ and may equal it. Conversely, any two admissible assignments of $[l,p-1]$ and $[p+1,r]$ at level $k$ together with $x_{p}=s_{k}$ form an admissible assignment of $[l,r]$ whose value is the sum in \eqref{eq:dp-E}. Induction on $r-l$ gives \eqref{eq:dp-D} and \eqref{eq:dp-E}, and $D(1,n,1)$ is the maximum of \eqref{eq:local-objective} over the whole grid, which equals $|\Win|\,\gap(a,b)$ by Lemma~\ref{lem:finite}.

For the cost, fix a level $k$. Two-dimensional prefix sums of the array $\varphi_{ij}(s_{k})$ over the band cost $O(n^{2})$ and afterwards return each $C(l,r,p,k)$ in constant time, so the $O(n^{2})$ intervals and the $O(n)$ positions per interval cost $O(n^{3})$ per level and $O(n^{3}m)$ over the $m$ levels. Level $k$ depends only on levels $k$ and $k+1$, so two $n\times n$ tables suffice when only the value is needed.
\end{proof}

Two remarks. The recurrence lets both sides of the split reach the value $s_{k}$ again, so ties among coordinates need no perturbation, and recovering a worst completion rather than its value raises the space to $O(n^{2}m)$. With integer codes every quantity in the programme is an integer and the sign of $\gap(a,b)$ is decided exactly, which is the property the certificates in Definition~\ref{def:certificates} need.

The programme is polynomial but not cheap, at about a millisecond per pair for $n=48$ and forty milliseconds for $n=144$ in our implementation. Most pairs do not need it. The next section develops a quadratic-time fast layer and characterises exactly when the feature-box relaxation is tight.

\section{Fast Layer and Algorithm}\label{sec:fast}

For every completion $x\in\Qbox$ and every window $(i,j)$ the minimum $\wmin{x}{i}{j}$ lies in the feature box $[\underline{z}_{ij},\overline{z}_{ij}]=[\wmin{\ell}{i}{j},\wmin{u}{i}{j}]$, and the two ends are attained by $\ell$ and $u$. Relaxing the coupling between windows gives two upper bounds on $\gap(a,b)$ that cost $O(|\Win|)$ time each. The independent bound treats the two distances separately,
\begin{equation}
\Uind=\frac{1}{|\Win|}\sum_{(i,j)\in\Win}\Bigl[\max_{v\in[\underline{z}_{ij},\overline{z}_{ij}]}|v-\alpha_{ij}|-\min_{v\in[\underline{z}_{ij},\overline{z}_{ij}]}|v-\beta_{ij}|\Bigr],
\label{eq:uind}
\end{equation}
and the relaxation bound maximises every term of \eqref{eq:gap} on its own,
\begin{equation}
\Brel=\frac{1}{|\Win|}\sum_{(i,j)\in\Win}\max_{v\in[\underline{z}_{ij},\overline{z}_{ij}]}\varphi_{ij}(v).
\label{eq:brel}
\end{equation}
Each maximum in \eqref{eq:brel} is attained at an end of the feature box or at a breakpoint $\alpha_{ij}$ or $\beta_{ij}$. Clearly $\gap(a,b)\le\Brel\le\Uind$, so a negative value of either bound certifies the pair. The bound $\Brel$ ignores that the same coordinates feed many windows. The following result says exactly when this loses nothing, and produces a completion as a by-product.

\subsection{When the relaxation is exact}

Call a window increasing when $\alpha_{ij}<\beta_{ij}$ and $\overline{z}_{ij}>\alpha_{ij}$, and set for such windows
\begin{equation}
\theta_{ij}=\min\{\overline{z}_{ij},\beta_{ij}\},\qquad
\xup_{t}=\max\Bigl\{\ell_{t},\ \max_{(i,j)\in\Win:\ i\le t\le j}\theta_{ij}\Bigr\},
\label{eq:xup}
\end{equation}
where the inner maximum runs over the increasing windows that contain $t$ and is $-\infty$ when there are none.

\begin{theorem}[Exactness of the relaxation]\label{thm:relax}
The completion $\xup$ lies in $\Qbox$, and $\gap(a,b)=\Brel$ if and only if $F_{a,b}(\xup)=\Brel$.
\end{theorem}
\begin{proof}
Feasibility holds because $\theta_{ij}\le\overline{z}_{ij}\le u_{t}$ for every $t$ in the window. Equality $\gap(a,b)=\Brel$ holds if and only if some completion attains every maximum in \eqref{eq:brel} at once, so we describe, for each window, the set of completions that attain its maximum. A window with $\alpha_{ij}=\beta_{ij}$ has $\varphi_{ij}\equiv 0$ and imposes nothing. For $\alpha_{ij}<\beta_{ij}$ the term $\varphi_{ij}$ is constant below $\alpha_{ij}$, increasing between $\alpha_{ij}$ and $\beta_{ij}$ and constant above $\beta_{ij}$. If $\overline{z}_{ij}\le\alpha_{ij}$ the term is constant on the feature box and imposes nothing. Otherwise its maximum is attained exactly when $\wmin{x}{i}{j}\ge\theta_{ij}$, that is, when $x_{t}\ge\theta_{ij}$ for every $t$ in the window. For $\alpha_{ij}>\beta_{ij}$ the term is nonincreasing on the whole line, so lowering any coordinate can only raise its value, and no condition needs to be recorded for such windows. The increasing windows impose a conjunction of lower bounds on coordinates, and $\xup$ is the coordinate-wise smallest point of $\Qbox$ that satisfies all of them.

If $F_{a,b}(\xup)=\Brel$ then $\gap(a,b)\ge F_{a,b}(\xup)=\Brel\ge\gap(a,b)$. Conversely, if $\gap(a,b)=\Brel$, let $x^{\ast}$ attain every maximum in \eqref{eq:brel}. Then $x^{\ast}$ satisfies all lower bounds, so $x^{\ast}\ge\xup$ coordinate-wise. For a window with $\alpha_{ij}>\beta_{ij}$ this gives $\wmin{\xup}{i}{j}\le\wmin{x^{\ast}}{i}{j}$ and hence $\varphi_{ij}(\wmin{\xup}{i}{j})\ge\varphi_{ij}(\wmin{x^{\ast}}{i}{j})$, which is the maximum, and the reverse inequality holds because $\xup$ is feasible. The increasing windows attain their maxima at $\xup$ by construction, and the remaining windows impose nothing, so $F_{a,b}(\xup)=\Brel$.
\end{proof}

The completion $\xup$ costs $O(n^{2})$ time through $G(i,j)=\max\{\theta_{ij},G(i-1,j),G(i,j+1)\}$ and $\xup_{t}=\max\{\ell_{t},G(t,t)\}$, where $\theta_{ij}$ is $-\infty$ for windows that are not increasing, out-of-range entries are $-\infty$, and the table is filled in order of decreasing window length, because every window that contains $[i,j]$ other than itself contains $[i-1,j]$ or $[i,j+1]$. Theorem~\ref{thm:relax} decides in quadratic time whether the cheap bound is already the exact gap.

\subsection{Witnesses and the certification procedure}

By Definition~\ref{def:certificates}, any feasible completion $z$ with $F_{a,b}(z)\ge 0$ refutes stability. The fast layer tries five such witnesses, the completion $\xup$, the two corners $\ell$ and $u$, the midpoint of the box, and the linear interpolation $\hat{x}$ of the observed slots clipped to the box. Each costs one evaluation of \eqref{eq:gap}. On integer-coded inputs all five are integer points of the box, the midpoint by integer division and the interpolation by rounding and clipping, so every $F_{a,b}(z)$ is evaluated in exact integer arithmetic and no witness verdict depends on floating point. The corner $\ell$ sets every gap to the lowest admissible value and pulls every window minimum that touches a gap to the bottom of its feature box, which is a strong perturbation at large $\tau$, whereas $\xup$ raises coordinates only where the increasing windows need it. Section~\ref{sec:experiments} shows that the two are complementary.

\begin{algorithm}[t]
\caption{PairCertify$(\Qbox,a,b,\tau)$}
\label{alg:pair}
\begin{algorithmic}[1]
\IF{$\tau=0$}
\RETURN CERTIFIED if \eqref{eq:l1-closed} is negative, else REFUTED
\ENDIF
\STATE compute $\Brel$ by \eqref{eq:brel} and $\xup$ by \eqref{eq:xup}
\IF{$\Brel<0$}
\RETURN CERTIFIED
\ENDIF
\FOR{$z$ in $\xup$, $\ell$, $u$, integer midpoint of $\Qbox$, $\hat{x}$}
\IF{$F_{a,b}(z)\ge 0$}
\RETURN REFUTED
\ENDIF
\ENDFOR
\STATE compute $\gap(a,b)$ by Theorem~\ref{thm:dp}
\RETURN CERTIFIED if $\gap(a,b)<0$, else REFUTED
\end{algorithmic}
\end{algorithm}

Algorithm~\ref{alg:pair} combines the pieces and is stated for integer-coded inputs. For real-valued boxes the witness checks use the unrounded midpoint and interpolation, which are feasible, and nothing else changes. Every CERTIFIED answer rests on a valid upper bound or on the exact gap, and every REFUTED answer rests on a feasible completion or on the exact gap, so the procedure is sound and complete for deciding whether $\gap(a,b)<0$, which is the strict stability of Definition~\ref{def:certificates}. It does not separate $\gap(a,b)=0$ from $\gap(a,b)>0$, since a witness with $F_{a,b}(z)=0$ already refutes strict stability. When $F_{a,b}(\xup)=\Brel$, Theorem~\ref{thm:relax} identifies the exact gap as $\Brel$, and Algorithm~\ref{alg:pair} returns the corresponding decision. The cost is $O(n^{2})$ per pair whenever the fast layer decides, and $O(n^{3}m)$ otherwise.

The guaranteed candidate set of Definition~\ref{def:certificates} needs a set of dominators. We rank the gallery by $\dist(\hat{x},\cdot)$ and take the first $\max\{3,k\}$ candidates as $\Dom$, then keep every $a\in\Gal$ that fewer than $k$ dominators defeat under Algorithm~\ref{alg:pair}, stopping the tests of a candidate as soon as $k$ dominators defeat it. The interpolation needs at least one observed slot, which every query in our experiments has. The nominal completion $\hat{x}$ chooses the dominators and serves as one witness, and Proposition~\ref{prop:candidate} holds for any choice, so a poor $\hat{x}$ can enlarge the set but never breaks its guarantee. For the label certificate we take the label $y$ of the first ranked candidate, the first three ranked candidates with label $y$, and test each of them against all candidates with a different label, hardest first in the nominal order, stopping at the first refutation. A candidate set costs at most $|\Dom|\,N$ calls of Algorithm~\ref{alg:pair}, and the label certificate at most three times the number of differently labelled candidates.

\section{Experiments}\label{sec:experiments}

\begin{table*}[!t]
\caption{Certification tests at $\tau=47$ and $k=1$ on 1\,000 queries per dataset with transplanted real masks. Bold marks the best value among sound tests. Complete refers to the pairwise decision of strict stability, not to minimality of the candidate set. DP-only timing uses the first 200 queries, other timings all 1\,000.}
\label{tab:main}
\centering
\scriptsize
\setlength{\tabcolsep}{3.5pt}
\begin{tabular}{@{}lcc rrrrr rrrrr rrrrr@{}}
\toprule
& & & \multicolumn{5}{c}{METR-LA} & \multicolumn{5}{c}{PEMS-BAY} & \multicolumn{5}{c}{PeMSD7-M} \\
\cmidrule(lr){4-8}\cmidrule(lr){9-13}\cmidrule(lr){14-18}
Test & Sound & Compl. & $|P_{1}|/N\downarrow$ & Ret.$\uparrow$ & Strict$\uparrow$ & False$\downarrow$ & ms$\downarrow$ & $|P_{1}|/N\downarrow$ & Ret.$\uparrow$ & Strict$\uparrow$ & False$\downarrow$ & ms$\downarrow$ & $|P_{1}|/N\downarrow$ & Ret.$\uparrow$ & Strict$\uparrow$ & False$\downarrow$ & ms$\downarrow$ \\
\midrule
IDB & \checkmark & $\times$ & 91.7 & 100.0 & 0.2 & 0.0 & 1.5 & 100.0 & 100.0 & 0.1 & 0.0 & 2.9 & 97.1 & 100.0 & 0.0 & 0.0 & 0.7 \\
FBR & \checkmark & $\times$ & 16.7 & 100.0 & 6.1 & 0.0 & \textbf{0.8} & 16.0 & 100.0 & 5.2 & 0.0 & \textbf{1.4} & 22.9 & 100.0 & 9.4 & 0.0 & \textbf{0.4} \\
\midrule
Fixed completions & $\times$ & -- & 13.1 & 98.6 & 7.3 & 26.5 & 0.4 & 13.9 & 98.2 & 6.5 & 20.2 & 0.7 & 17.1 & 99.5 & 11.2 & 14.6 & 0.2 \\
MC sampling, $K=100$ & $\times$ & -- & 13.6 & 99.3 & 6.9 & 18.7 & 5.2 & 14.0 & 98.8 & 6.2 & 15.2 & 8.4 & 20.0 & 99.9 & 10.9 & 7.6 & 4.4 \\
MC sampling, $K=1000$ & $\times$ & -- & 14.3 & 99.6 & 6.8 & 15.9 & 43.7 & 14.6 & 99.1 & 6.0 & 13.3 & 77.4 & 20.0 & 99.9 & 10.8 & 5.8 & 19.7 \\
\midrule
\method{} w/o fast layer & \checkmark & \checkmark & 15.4 & 100.0 & 6.8 & 0.0 & 175.6 & 15.3 & 100.0 & 5.9 & 0.0 & 265.3 & 20.0 & 100.0 & 10.3 & 0.0 & 79.5 \\
\method{} (ours) & \checkmark & \checkmark & \textbf{15.4} & 100.0 & \textbf{6.8} & 0.0 & 3.2 & \textbf{15.3} & 100.0 & \textbf{5.9} & 0.0 & 4.2 & \textbf{20.0} & 100.0 & \textbf{10.3} & 0.0 & 0.9 \\
\midrule
\multicolumn{18}{@{}l}{\textit{Ablations of \method{}}} \\
one dominator & \checkmark & \checkmark & 18.8 & 100.0 & 6.8 & 0.0 & 1.8 & 18.5 & 100.0 & 5.9 & 0.0 & 2.0 & 25.7 & 100.0 & 10.3 & 0.0 & 0.3 \\
ten dominators & \checkmark & \checkmark & 12.0 & 100.0 & 6.8 & 0.0 & 6.4 & 12.5 & 100.0 & 5.9 & 0.0 & 7.7 & 17.1 & 100.0 & 10.3 & 0.0 & 2.5 \\
midpoint completion & \checkmark & \checkmark & 18.1 & 100.0 & 6.8 & 0.0 & 4.9 & 33.8 & 100.0 & 5.9 & 0.0 & 5.9 & 22.9 & 100.0 & 10.3 & 0.0 & 1.8 \\
$\tau=0$, the L1 end & \checkmark & \checkmark & 6.1 & 100.0 & 22.6 & 0.0 & 0.1 & 2.8 & 100.0 & 27.7 & 0.0 & 0.2 & 5.7 & 100.0 & 33.3 & 0.0 & 0.1 \\
\bottomrule
\end{tabular}
\end{table*}

\begin{figure*}[!t]
\centering
\includegraphics[width=0.92\textwidth]{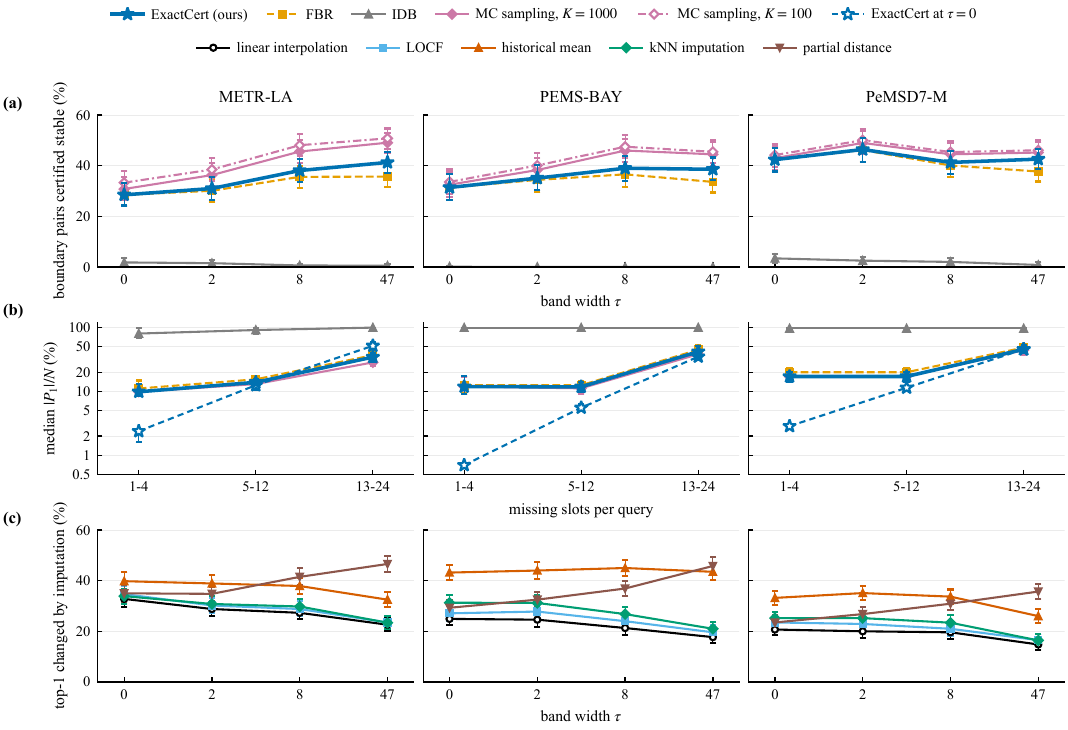}
\caption{Certification results on three traffic-sensor datasets with transplanted real masks. (a) Share of top-5 boundary pairs that each test certifies as stable, for growing band width $\tau$. \method{} decides $\gap<0$ exactly, so its curve is the ground truth. (b) Median size of the guaranteed candidate set as a share of the gallery, by the number of missing slots out of 48, at $\tau=47$, with \method{} at $\tau=0$ shown as the L1 end of the family. Sets produced by \method{}, FBR and IDB contain every true nearest neighbour whenever the true query lies in the box. (c) Share of queries whose top-1 changes under each missing-data practice, against $\tau$. $\tau$ takes the values 0, 2, 8 and 47 on an evenly spaced axis, where $\tau=0$ is the L1 distance and $\tau=47$ the full window-minimum distance. Error bars are 95\% intervals from a bootstrap over sensors.}
\label{fig:main}
\end{figure*}

\subsection{Setup}\label{sec:setup}

\textbf{Data.} Table~\ref{tab:datasets} lists the datasets. METR-LA, PEMS-BAY and PeMSD7-M are five-minute loop-detector speeds from the releases of \cite{li2018dcrnn} and \cite{yu2018stgcn}, with zeros treated as missing as in those releases. Together they hold 760 sensors and 93\,165 sensor-days. The three labelled UCR sets \cite{dau2019ucr} add classification tasks. DodgerLoopDay \cite{ihler2006adaptive} contains real gaps and is aggregated to fifteen-minute block sums, so that a bin with missing points becomes a training-bounded interval. ItalyPowerDemand and Wafer are complete and are used with controlled deletion.

\begin{table}[t]
\caption{Datasets. Series and test are sensors and days for the traffic sets and training and test series for the UCR sets. Gallery is the median same-sensor gallery or the number of complete training series.}
\label{tab:datasets}
\centering
\scriptsize
\setlength{\tabcolsep}{3.5pt}
\begin{tabular}{@{}lrrrrr@{}}
\toprule
Dataset & Sens./train & Days/test & $n$ & Missing & Gallery \\
\midrule
METR-LA & 207 & 119 & 48 & 8.1\% & 82 \\
PEMS-BAY & 325 & 180 & 48 & $<$0.01\% & 144 \\
PeMSD7-M & 228 & 44 & 48 & 0 & 35 \\
\midrule
DodgerLoopDay & 78 & 80 & 96 & real & 67 \\
ItalyPowerDemand & 67 & 1\,029 & 24 & 0 & 67 \\
Wafer & 1\,000 & 6\,164 & 152 & 0 & 1\,000 \\
\bottomrule
\end{tabular}
\end{table}

\textbf{Representation and boxes.} Speeds are coded as integers at one eighth of a mile per hour, so all certificates are decided in exact integer arithmetic. Codes are rounded to the nearest integer, and every guarantee and every ground truth below refers to the coded data. Each traffic dataset is split in time, the first 80\% of days form the history and the last 20\% the test period. A query is one test day of one sensor in the afternoon window, and its gallery holds the complete history days of the same sensor, with median size 82, 144 and 35 for the three sets. A missing slot receives a history-bounded box from zero to the largest speed the sensor recorded in its history, a convention rather than a physical limit. The true value then lies inside the box for 100\%, 99.7\% and 100\% of the queries used below. To have ground truth we take complete test days and transplant real gap patterns onto them. The pattern pool holds the 2\,773 partially missing METR-LA windows with one to 24 missing slots, median six, and 78\% of them are a single contiguous block. Each experiment uses 1\,000 such queries per dataset with a fixed seed. The four band widths share the same queries, so the main experiment evaluates 3\,000 queries at four band widths, 12\,000 query and band-width combinations. The 621 METR-LA test days with real gaps are evaluated separately, without ground truth.

\textbf{Band widths and baselines.} We report $\tau\in\{0,2,8,47\}$, from L1 to the full band, with $\tau=47$ as the main setting. Certification tests are compared in three groups. Sound but incomplete tests are the independent bound IDB of \eqref{eq:uind} and the relaxation FBR of \eqref{eq:brel}. Heuristics without guarantees declare a pair stable only when $F_{a,b}(z)<0$ for every inspected completion $z$, so an inspected tie also refutes strict stability. They are a fixed-completion check that inspects the two corners $\ell$ and $u$, the midpoint of the box and the linear interpolation of the observed slots, and Monte Carlo sampling with $K=100$ or $K=1000$ uniform completions of the box plus the same four. Neither uses $\xup$, which depends on the pair. \method{} with the fast layer removed, that is, the dynamic programme on every pair, is the ablation of Section~\ref{sec:fast}. Missing-data practices are linear interpolation, last observation carried forward, the historical mean of the slot, kNN imputation with five neighbours, and partial distance on the fully observed windows.

\textbf{Metrics.} Top-1 changed means that the candidate returned for the imputed query is not among the candidates at the minimum true distance, so ties are not counted as changes. Set size is the median of $|P^{1}|/N$ over the queries. Retention is the share of queries, among those whose truth lies inside the box, in which every candidate at the minimum true distance is kept in $P^{1}$. Strict is the share of queries with $|P^{1}|=1$. False is the share of certificates issued on the top-5 boundary pairs that are wrong, where a boundary pair joins the fourth or fifth and the sixth or seventh candidate of the nominal ranking. Time is the median milliseconds per query for building $P^{1}$, after Numba warm-up and excluding the precomputation of gallery features. Quality metrics and the times of every row except one use all 1\,000 queries per dataset, and the programme-only row was timed on the first 200 of them.

\textbf{Protocol and uncertainty.} All thresholds and budgets were fixed on a pseudo-split inside the history period before any test-period result was seen. The method is deterministic, so the uncertainty reported is a bootstrap over sensors with 1\,000 resamples, and for the deletion experiments the spread over three seeds. For the percentage-valued metrics in Table~\ref{tab:main} the bootstrap intervals have half-width at most 3 points for every entry except the set size of IDB on METR-LA, where it is 9 points. The kernel was verified against brute-force enumeration in exact rational arithmetic on 2\,350 random instances, including a refined grid between critical values for Lemma~\ref{lem:finite} and the equivalence of Theorem~\ref{thm:relax}. Everything runs single-threaded in NumPy and Numba on one Mac M5 laptop.

\subsection{Imputation changes retrieval}

Fig.~\ref{fig:main}(c) reports how often the top-1 returned after imputation differs from the top-1 of the true query. With linear interpolation and the L1 distance this happens for 32.8\%, 24.9\% and 20.7\% of the queries on the three sets, and 22.6\%, 17.7\% and 14.8\% at the full band. The historical mean is the worst practice at $\tau=0$, at 33\% to 43\%, and partial distance, which discards every window that touches a gap, degrades as $\tau$ grows and is the worst at the full band with 36\% to 47\%, because long windows almost always touch the gap. The certificate separates the reliable answers from the rest. For linear interpolation, the top-1 change rate is 0\% among the queries that \method{} certifies strictly, on every dataset and at every $\tau$, while among the remaining queries it is 42.4\%, 34.4\% and 31.0\% at $\tau=0$ and 24.2\%, 18.8\% and 16.5\% at $\tau=47$.

\subsection{Certificates against cheaper tests}

Table~\ref{tab:main} is the main comparison at $\tau=47$ and $k=1$, and Fig.~\ref{fig:main}(a) and (b) show the same tests across $\tau$ and across the amount missing. Three findings.

\textbf{Sound tests are sound, and \method{} is the tightest of them.} On queries whose true series lies in the box, all sound tests retain every true nearest neighbour, as Proposition~\ref{prop:candidate} guarantees, and no strictly certified top-1 was wrong on any of the 1\,000 queries per dataset. Table~\ref{tab:main} shows how much each prunes. IDB prunes almost nothing, and \method{} reduces the median candidate-set fraction by 0.7 to 2.9 points relative to FBR. The difference is larger at the pair level, on a challenging set of top-5 boundary pairs. Of the pairs that are exactly stable, FBR certifies 86.4\%, 87.0\% and 88.5\% and IDB 1.2\%, 0.0\% and 1.9\%, while \method{} certifies all of them by construction, and Fig.~\ref{fig:main}(a) shows the growth of this margin with $\tau$. At $\tau=0$ FBR coincides with the exact gap, as Proposition~\ref{prop:l1} predicts, while IDB stays loose because it maximises and minimises the two distances at different completions.

\textbf{Heuristics look better and are wrong.} The fixed-completion check and Monte Carlo sampling produce smaller sets than \method{} by accepting additional, unverified dominations, and some of these exclusions remove true nearest neighbours. The fixed-completion check loses one in 14, 18 and 5 of the 1\,000 queries, sampling with $K=100$ in 7, 12 and 1, and sampling with $K=1000$ in 4, 9 and 1. On the boundary pairs, 26.5\%, 20.2\% and 14.6\% of the fixed-completion certificates and 15.9\%, 13.3\% and 5.8\% of the $K=1000$ certificates are false. Sampling is the slowest of the bounds and heuristics, because every candidate costs $K$ evaluations of the distance, although it stays faster than the programme on every pair.

\textbf{Exactness costs little.} On the top-5 boundary pairs the fast layer decides 90.2\%, 92.0\% and 92.7\% of them without the dynamic programme, and the relaxation is exactly tight on 34.5\%, 48.3\% and 53.2\% of them. A candidate set costs 3.2, 4.2 and 0.9 milliseconds on the same-sensor galleries, against 176, 265 and 80 milliseconds for the dynamic programme alone, with identical sets. Ties, $\gap(a,b)=0$, occur on 0.1\% to 2.3\% of the boundary pairs.

\subsection{What the band width trades}

\begin{figure*}[!t]
\centering
\includegraphics[width=0.92\textwidth]{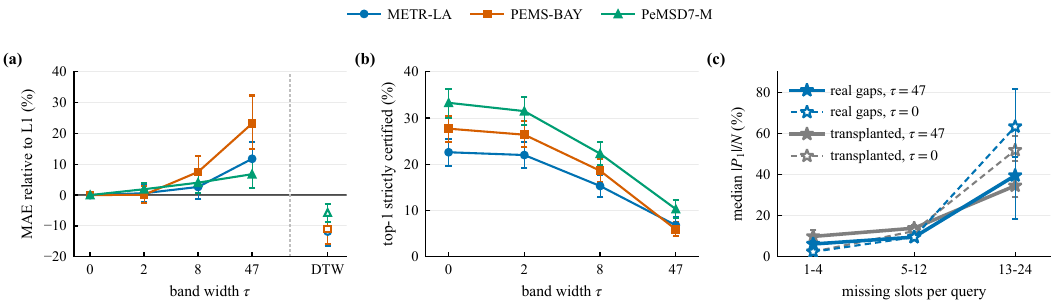}
\caption{Effects of band width on forecasting accuracy and certified retrieval. Panels (b) and (c) use \method{}. (a) Change of the 5-NN next-hour MAE relative to the L1 distance on each dataset, with hollow markers giving the same quantity for tuned DTW. (b) Share of test queries whose top-1 is strictly certified. (c) Median size of the guaranteed candidate set on METR-LA by the number of missing slots out of 48, for test days with real gaps and for complete test days with transplanted real masks. The $\tau$ axis and the error bars are as in Fig.~\ref{fig:main}. The real-gap stratum with 13 to 24 missing slots contains 45 queries and the transplanted one 184.}
\label{fig:tradeoff}
\end{figure*}

The forecast experiment uses complete test days as queries, so it compares distances rather than distances combined with an imputer. For each query the five nearest history days of the same sensor under the given distance are retrieved, the prediction is the average of their next-hour mean speed, and the error is the mean absolute error in miles per hour. DTW uses a Sakoe--Chiba band whose width fraction was chosen from $\{0, 0.05, 0.1, 0.2, 1\}$ on the history pseudo-split, with ties going to the smaller band, which selected 1 on METR-LA and 0.2 on the other two.

Fig.~\ref{fig:tradeoff} shows what the band width trades. Accuracy falls. The forecast error grows from the L1 distance to the full band by 11.8\%, 23.2\% and 6.8\%, while tuned DTW improves on L1 by 11.8\%, 11.0\% and 5.7\% and remains outside the certifiable family. Strict certification falls too, from 22.6\%, 27.7\% and 33.3\% of the queries at $\tau=0$ to 6.8\%, 5.9\% and 10.3\% at $\tau=47$. Robustness to imputation rises, as the change rates above show. The regime with many missing slots is where the full band pays. On METR-LA test days with real gaps and 13 to 24 missing slots, the median guaranteed set is 63.3\% of the gallery under L1 and 39.5\% under the full band, and the transplanted masks give the same picture on METR-LA with 51.9\% against 34.5\%. Fig.~\ref{fig:main}(b) shows that PEMS-BAY reverses the order and PeMSD7-M shows no difference, so the crossover is dataset dependent. In our experiments, L1 is cheaper and yields smaller candidate sets when few slots are missing.

\subsection{Efficiency and scale}

Table~\ref{tab:efficiency} has two parts. The first varies the window length on PEMS-BAY from $n=24$ to $n=144$ at full band, on 200 boundary pairs per length with the same real gap patterns placed inside the window. The dynamic programme grows from 0.16 to 40 milliseconds per pair, the fast layer from 0.009 to 0.24 milliseconds, and the complete method from 0.02 to 5.4 milliseconds on average because it runs the programme on 2.5\% to 13.5\% of the pairs. The growth is slower than the $O(n^{3}m)$ bound since the admissibility test inside \eqref{eq:dp-E} prunes most positions. The second part uses every complete history day of every sensor as one gallery, with 16\,455, 46\,719 and 7\,980 candidates. At $\tau=47$ a guaranteed set then takes 0.31, 0.71 and 0.15 seconds per query, against 36, 76 and 17 seconds for the programme alone. On the ten queries per dataset that were run both ways the speedup is 82, 100 and 110 with identical answers. On average 158, 310 and 77 of the 23\,674, 66\,392 and 11\,675 pair tests per query reach the programme, about ten million pair tests over the 300 queries. The relaxation alone takes 0.14, 0.40 and 0.07 seconds, so exactness costs about a factor of two at network scale.

\begin{table}[t]
\caption{Cost of certification. Top: milliseconds per pair against the window length on PEMS-BAY at full band, where Fast and DP report median time, DP p95 the 95th percentile, and Full mean time. $m$ reports the median critical-set size over the sampled pairs. Bottom: seconds per query on the network-wide gallery at $\tau=47$, medians over 100 queries per dataset for \method{} and FBR and over the ten queries on which the programme alone was run. Speedup is the ratio of unrounded medians on those same ten queries, on which \method{} takes 0.44, 0.76 and 0.16 seconds.}
\label{tab:efficiency}
\centering
\scriptsize
\setlength{\tabcolsep}{3.5pt}
\begin{tabular}{@{}rrrrrrr@{}}
\toprule
$n$ & $m$ & Fast & DP & DP p95 & Full & DP needed \\
\midrule
24 & 33 & 0.009 & 0.16 & 0.25 & 0.02 & 2.5\% \\
48 & 52 & 0.016 & 1.19 & 2.20 & 0.14 & 9.5\% \\
96 & 76 & 0.057 & 9.74 & 19.63 & 1.38 & 13.0\% \\
144 & 124 & 0.235 & 39.96 & 65.10 & 5.37 & 13.5\% \\
\bottomrule
\end{tabular}

\vspace{1ex}

\begin{tabular}{@{}lrrrrr@{}}
\toprule
Dataset & $N$ & \method{} & FBR & DP only & Speedup \\
\midrule
METR-LA & 16\,455 & 0.31 & 0.14 & 36.2 & 82$\times$ \\
PEMS-BAY & 46\,719 & 0.71 & 0.40 & 75.8 & 100$\times$ \\
PeMSD7-M & 7\,980 & 0.15 & 0.07 & 17.2 & 110$\times$ \\
\bottomrule
\end{tabular}
\end{table}

\subsection{Label invariance on public benchmarks}

\textbf{Protocol.} ItalyPowerDemand and Wafer are coded by $c(x)=\operatorname{round}\bigl[4096\cdot\tfrac12\bigl(1+x/(s+|x|)\bigr)\bigr]$ with $s$ the median absolute training value, an order-preserving map to integer codes in $[0,4096]$. The transform fixes the representation on which both nominal retrieval and certification are defined, and nearest-neighbour rankings in the original units are not assumed to be preserved. DodgerLoopDay counts are summed over blocks of three five-minute points, and a block with $k$ missing points receives the interval from the sum of its observed points to that sum plus $k$ times the largest single count in the training set, which is 90. Its gallery is the 67 training series without gaps, and weekend means the official classes 1 and 7, Sunday and Saturday. Deletion removes one contiguous block covering the stated share of the slots, with a uniformly random start, and replaces it by the box $[0,4096]$ for the two coded sets and $[0,270]$ for DodgerLoopDay. On the 77 complete DodgerLoopDay test series the largest block sum is 181, so the deletion box covers every masked value, and the coded sets cover theirs by construction. Queries are all 1\,029 ItalyPowerDemand test series, the 77 complete DodgerLoopDay test series, and 100 Wafer test series per class, and each deletion level is repeated with three seeds.

Table~\ref{tab:labels} applies Proposition~\ref{prop:label} with contiguous block deletion of 25\% of the slots on the three UCR sets. On ItalyPowerDemand the label of 53.4\% of the test series is certified invariant at $\tau=0$, and the nominal 1-NN is correct on 98.1\% of the certified series against 77.5\% of the others. On DodgerLoopDay with weekday against weekend labels the numbers are 49.8\% certified and 100\% against 89.7\%, and on Wafer, with balanced accuracy, 64.8\% certified and 100\% against 95.2\%. In no cell does a certified label change under the true series, which is the guarantee. Deleting 10\% of the slots certifies 92\% to 96\% of the series on ItalyPowerDemand and Wafer at $\tau=0$, and deleting 50\% certifies 12\% on ItalyPowerDemand and none on Wafer. The original seven-class task is challenging in our 1-NN experiments, with 59.7\% accuracy on complete series, and certifies 3\% of the queries at 25\% deletion. Wafer certifies its abnormal class less often than its normal class, 41\% to 50\% against 80\% to 86\% at $\tau=0$.

\begin{table}[t]
\caption{1-NN label invariance under 25\% contiguous deletion, mean and population standard deviation over three seeds, in percent. Wafer uses balanced accuracy on 100 queries per class.}
\label{tab:labels}
\centering
\scriptsize
\setlength{\tabcolsep}{4pt}
\begin{tabular}{@{}rrrrr@{}}
\toprule
$\tau$ & Certified$\uparrow$ & Acc.\ certified & Acc.\ uncertified & Acc.\ all \\
\midrule
\multicolumn{5}{@{}l}{\textit{ItalyPowerDemand}, $n=24$} \\
0 & 53.4 $\pm$ 1.0 & 98.1 $\pm$ 0.1 & 77.5 $\pm$ 2.1 & 88.5 $\pm$ 0.8 \\
8 & 30.9 $\pm$ 0.7 & 97.3 $\pm$ 0.6 & 79.0 $\pm$ 0.8 & 84.6 $\pm$ 0.4 \\
23 & 21.0 $\pm$ 1.0 & 97.5 $\pm$ 0.5 & 79.1 $\pm$ 0.6 & 83.0 $\pm$ 0.6 \\
\midrule
\multicolumn{5}{@{}l}{\textit{DodgerLoopDay}, weekday against weekend, $n=96$} \\
0 & 49.8 $\pm$ 2.2 & 100.0 $\pm$ 0.0 & 89.7 $\pm$ 2.1 & 94.8 $\pm$ 1.1 \\
8 & 43.7 $\pm$ 7.1 & 97.8 $\pm$ 1.6 & 89.9 $\pm$ 1.4 & 93.5 $\pm$ 1.1 \\
95 & 8.2 $\pm$ 2.4 & 96.3 $\pm$ 5.2 & 95.3 $\pm$ 0.6 & 95.2 $\pm$ 0.6 \\
\midrule
\multicolumn{5}{@{}l}{\textit{Wafer}, balanced accuracy, $n=152$} \\
0 & 64.8 $\pm$ 1.0 & 100.0 $\pm$ 0.0 & 95.2 $\pm$ 1.8 & 97.8 $\pm$ 0.5 \\
8 & 64.8 $\pm$ 3.7 & 100.0 $\pm$ 0.0 & 95.3 $\pm$ 1.9 & 98.0 $\pm$ 0.4 \\
151 & 30.3 $\pm$ 2.6 & 100.0 $\pm$ 0.0 & 92.4 $\pm$ 1.3 & 94.5 $\pm$ 1.1 \\
\bottomrule
\end{tabular}
\end{table}

\subsection{Ablations}

The last block of Table~\ref{tab:main} varies one choice of Section~\ref{sec:fast} at a time. More dominators shrink the set with diminishing returns and rising cost, from 18.8\% and 1.8 milliseconds per query with one dominator on METR-LA to 12.0\% and 6.4 milliseconds with ten, while strict certification does not move at all, because it is decided by the first dominator. The nominal completion barely matters as long as it is sensible, since linear interpolation, last observation carried forward, historical mean and kNN give median sets within 1.4 points of each other on the first two datasets, while the midpoint of the box, which sits far from any real speed, gives 18.1\% and 33.8\%. The witnesses are complementary. Among all pair tests made while building the candidate sets at $\tau=47$, a population dominated by easy pairs unlike the boundary pairs above, the bound leaves about half undecided, and of those $\xup$ alone leaves the programme to 28.4\%, 14.7\% and 17.6\%, the four fixed completions alone to 7.3\%, 3.9\% and 6.0\%, and all five together to 4.6\%, 2.5\% and 3.6\%. At $\tau=8$ the order of the two witness groups reverses, which is why Algorithm~\ref{alg:pair} keeps both.

\section{Limitations and Conclusion}\label{sec:limitations}

We give a complete decision procedure for strict pairwise stability under window-minimum distances, combining a finite critical-value reduction, an exact interval dynamic programme and a quadratic-time fast layer. Pairwise decisions yield sound candidate sets and sufficient label-invariance certificates.

The guarantees hold over the specified box and apply to the true query only when it lies inside, which held for all but 0.3\% of the PEMS-BAY queries. The framework requires a complete gallery and fixed time positions. Broad boxes limit certification coverage, and tighter boxes impose stronger assumptions on the missing values. The programme takes $O(n^{3}m)$ time with $m\le 4n$, and its practicality is shown only at the tested window lengths and gallery sizes. The main retrieval experiments use traffic-speed data, and the three label benchmarks provide limited evidence beyond that setting. Future work includes incomplete galleries, window minima on trees or graphs, and retrieval by worst-case regret.

\bibliographystyle{IEEEtran}
\bibliography{refs}

\end{document}